\documentclass[12pt]{amsart}
\usepackage{amsmath,amssymb,amsthm,latexsym,graphicx,cite}
\newtheorem{theorem}{Theorem}

\newtheorem{remark}[theorem]{Remark}
\newtheorem{corollary}[theorem]{Corollary}

\newtheorem{proposition}[theorem]{Proposition}
\def\C{\Bbb{C}}
\def\R{\Bbb{R}}
\def\Z{\Bbb{Z}}
\def\T{\Bbb{T}}

\usepackage{lineno}
\title{Dispersion relations and spectra of periodically perforated structures}
\author{Peter Kuchment}
\address{Mathematics Department, Texas A\&M University. College Station, TX 77843-3368, USA}
\email{kuchment@math.tamu.edu}
\author{Jari Taskinen}
\address{Department of Mathematics and Statistics, University of Helsinki}
\email{jari.taskinen@helsinki.fi}
\dedicatory{Dedicated to Professor Shmuel Agmon, a great mathematician and inspiration}
\thanks{P.K. acknowledges support of the NSF DMS grants \# 1517938 and \# 2007408. J.~T. was partially supported by a research grant from the Faculty of Science of the University of Helsinki}
\begin{document}
\begin{abstract}
We establish absolute continuity of the spectrum of a periodic Schr\"odiner operator in $\R^n$ with periodic perforations. We also prove analytic dependece of the dispersion relation on the shape of the perforation.
\end{abstract}
\date{\today}
\maketitle

\section*{Introduction}
The theory of periodic partial differential equations is at least a century old (see, e.g., \cite{Bril,KuchBAMS,KuchBook,RS,Skri}), but is still widely active, mostly due to its importance for various areas of mathematical physics, such as solid state physics, photonic crystal theory, topological insulators theory, and nano-science, to name just a few (see \cite{AM,Bernevig,Bril,Bush,Do,DoKuch,FW,Joan,Korot,KuchBAMS,KucPBG,KucVainb,KucPost}. Periodicity is usually introduced by crystalline structure of materials, or its optical analogs. Among the topics being discussed one can mention as examples analytic structure of the corresponding dispersion relations and spectral structure of periodic operators. Another option is to study media with periodically modulated shapes, or multi-periodic perforations, see e.g. \cite{card,ferrar,TItal,NRT} and references therein. The latter is the topic we address here. We will be concentrating on the following two questions: dependence of dispersion relations (and thus spectra) on the shape of perforation and absolute continuity of the spectrum. The former has been considered for instance for the case of circular perforations with respect to a varying radius in $\R^2$ \cite{TItal}. The latter is a well known, much studied, but still not completely finished topic (see \cite{KuchBAMS} for references and discussion), when periodicity arises due to periodic coefficients, rather periodic perforations.

The goal of this article is to show that a combination of several known for other situations approaches and results enable one to establish with ease some very general properties of interest for periodically perforated domains in any dimension.

The text is structured as follows: Section \ref{S:perturb} briefly refers to some powerful (and not always that well known) techniques of domain perturbations \cite{CH,Garab,Henry,IvKaKr}, following \cite{Henry}. An important for the further discussions Theorem \ref{T:domainanalit} on analytic dependence on domain perturbations is established. The next Section \ref{S:perf} describes the perforated geometry of interest. Theorem \ref{T:disp} of Section \ref{S:Disp} establishes a very general analyticity with respect to shape variations result. 
Section \ref{S:ac} contains the proof of Theorem \ref{T:ac} on absolute continuity of spectra of such structures (i.e., impossibility of creating a bound state by periodic perforations), allowing also for presence of periodic electric potential. Section \ref{S:remarks} contains some additional remarks. It is followed by the Acknowledgments section.

\section{Domain perturbations}\label{S:perturb}
Let $\Omega\subset\R^n$ be a smooth bounded domain and $l(x,z,D)$ be a linear elliptic partial differential expression of order $m$ with ``sufficiently nice'' coefficients\footnote{``Sufficiently nice'' means that the only property that we need is that the boundary value problem produces a Fredholm operator from the Sobolev space $H^m$ in $\Omega$ with the corresponding boundary conditions to $L_2(\Omega)$.} defined in a neighborhood of the closure of $\Omega$, where we use the standard PDE notation $D$ for $\dfrac{1}{i}\dfrac{\partial}{\partial x}$. The coefficients are allowed to depend analytically on a parameter $z$ in a domain $\mathcal{C}\subset \C^l$ for some integer $l\geq 1$, or a complex analytic space, or even domain in a complex Banach space $E$. We assume that boundary conditions $Bu|_\Gamma=0$ are imposed on $\Gamma:=\partial \Omega$ that lead to an elliptic boundary value problem for the operator $L(z)$ acting as $l(x,z,D)$ in $\Omega$. The coefficients of the boundary operators are also allowed to depend analytically on $z$.

Let us denote by $H^m_z$ the closed subspace of $H^m(\Omega)$ consisting of all functions satisfying the boundary conditions $B(z)u|_\Gamma=0$.

The following claims are standard:
\begin{proposition}\label{P:bundle}\indent
\begin{enumerate}
\item There exists, locally in $z$, a projector $P(z):H^2(\Omega)\mapsto H^2_z$ analytically dependent on $z$, and thus its range forms an analytic subbundle\footnote{If the domain $\mathcal{C}$ is holomorphically convex (or is an abstracts complex Stein space), then an analytic projector exists globally in $z$ (see \cite{Shubin,ZKKP}), but this is not needed for our results.}
    \begin{equation}
    \mathcal{F}:=\bigsqcup\limits_{z\in\mathcal{C}} H^m_z
    \end{equation}
in the trivial bundle $\mathcal{C}\times H^m(\Omega)$ over $\mathcal{C}$.
\item The operator $L(z)$ produces a Fredholm morphism between bundles $\mathcal{F}$ and $\mathcal{C}\times L_2(\Omega)$.
\item The ``dispersion relation''
\begin{equation}
\mathcal{D}:=\{(\lambda,z)\in \C\times\mathcal{C}\,|\, L(z)u=\lambda u \mbox{ has a non-zero solution}\}
\end{equation}
is an analytic subset in $\C\times\mathcal{C}$. It is principal (i.e., is defined as the set of zeros of a \underline{single} analytic function $f(\lambda,z)$) if the Fredholm index of the operator is equal to zero.
\end{enumerate}
\end{proposition}

Indeed, local existence of an analytic projector is a simple exercise (see, e.g. \cite{ZKKP,Shubin}. The main notions and results concerning analytic Banach bundles and Fredholm morphisms that explain the rest of the first two claims can be found in \cite{ZKKP}. The last statement of the theorem follows from \cite[Theorem 4.11 and its Corollary]{ZKKP}.

We now show how domain variations fit into this scheme. Since analyticity is a local property, we will be looking at small shape variations only.
A convenient (although over-determined and often neglected) way to parameterize domain variations is by varying its natural embedding into the ambient space, as opposed to parameterizations by normal perturbations of the boundary. Thus, let $\Omega\in\R^n$ be a bounded domain with the smooth boundary $\Gamma$.
 \begin{figure}[ht!]
 \centering
  \includegraphics[scale=1]{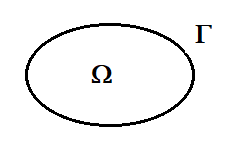}
\caption{A smooth domain $\Omega\subset \R^n$ and its boundary $\Gamma$.}\label{F:domain}
\end{figure}
We denote by $I_\Omega$ its natural embedding into $\R^n$ and consider the Banach space $C^m(\Omega, \R^n)$ of $m$ times continuously (and uniformly) differentiable mappings from $\Omega$ to $\R^n$. Then $I_\Omega\in C^m(\Omega, \R^n)$.
Consider  $h\in C^m(\Omega, \R^n)$ such that $\|(h-I_\Omega)\|$ is sufficiently small and denote the corresponding small ball by $R\subset C^m(\Omega, \R^n)$. Then $h$ is still a diffeomorphic embedding of $\Omega$ into $\R^n$. The domains $\Omega_h:=h(\Omega)$ are ``small perturbations'' of domain $\Omega$, and their boundaries $\Gamma_h:=h(\Gamma)$ are small perturbations of $\Gamma$. Let us now define the operator $L(z,h)$ acting as $l(x,z,D)$ on the domain $\Omega_h$ with elliptic boundary conditions $B(z,h)u|_{\Gamma_h}=0$, where the boundary operators $B$ may depend analytically on $h$.

The diffeomorphisms $h$ enable us to pull-back the BVPs from the domains $\Omega_h$ back to $\Omega$, ending up with a new family of elliptic operators, which we will call $M(h,z)$. A simple calculation (see, e.g., \cite[bottom of page 20]{Henry}) shows that the coefficients of the operator $M$ depend \textbf{analytically}\footnote{This analyticity does not have anything to do with smoothness of the coefficients of the operator or the surface $\Gamma$.} on $h\in R$. Thus, the following result is just a corollary of the proposition \ref{P:bundle}:
\begin{theorem}\label{T:domainanalit}
The ``dispersion relation''
\begin{equation}
\mathcal{D}:=\{(\lambda,z,h)\in \C\times\mathcal{C}\times R\,|\, L(z,h)u=\lambda u \mbox{ has a non-zero solution}\}
\end{equation}
is an analytic subset in $\C\times\mathcal{C}\times R$. It is principal (i.e., is defined as the set of zeros of a single analytic function $f(\lambda, z,h)$) if the Fredholm index of the operator is equal to zero.
\end{theorem}

One simple consequence of this result is:
\begin{corollary}\label{C;simpleeigenv}
If $\lambda$ is a simple eigenvalue of the operator $L(z_0)$, then it extends analytically to a simple eigenvalue $\lambda(z,h)$ of $L(z,h)$ for sufficiently small $|z-z_0|$ and $\|h-I_\Omega\|$.
\end{corollary}

\begin{remark}\label{R:homot}
In particular, the results of Theorem \ref{T:domainanalit} and Corollary \ref{C;simpleeigenv} apply to the simplest case of a spherical domain of changing radius \cite{TItal}, or in fact to homothetic perturbation of a star-shaped domain.
\end{remark}

\section{Perforated geometry}\label{S:perf}

Our main goal here is to consider a periodically perforated medium of the following kind: The domain of consideration is the space $\R^n$ with a $\Z^n$-periodic arrangement of non-overlapping smooth contractible bounded domains (``holes'') removed (see the shaded domain $W$ in  Fig. \ref{F:struct}).
\begin{figure}[ht!]
  \centering
  \includegraphics[scale=1.3]{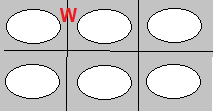}
\caption{The space $\R^n$ being periodically perforated by contractible smooth holes.}\label{F:struct}
\end{figure}
Consider in $W$ a periodic elliptic operator $L(x,D)$ of order $m$ with elliptic conditions on the boundary of $W$ (i.e., on the boundaries of the holes) and with sufficiently nice coefficients for the Fredholm property between the appropriate Sobolev spaces to hold. We are interested in its spectrum, and thus first of all in its dispersion relation (see, e.g. \cite{KuchBAMS,KuchBook} and references therein)
\begin{equation}\label{E:disp}
\mathcal{D}:=\{(k,\lambda)\in \C^n\times \C\,|\, \exists u\neq 0 \,\, \Z^n\mbox{-periodic, } L(x, D+k)u=\lambda u\}.
\end{equation}
It is well known (see \cite{KuchBAMS} and references therein) that the spectrum of the operator coincides with the projection of $\mathcal{D}$ on the $\lambda$-plane. Also, according to Proposition \ref{P:bundle}, it is an analytic set.

We are interested now in dependence of this picture on variations of the shape of the holes.

\section{Analyticity of the extended dispersion relation}\label{S:Disp}

Since in (\ref{E:disp}) only $\Z^n$-periodic functions are of interest, we can concentrate on a single fundamental domain $\Omega$ folded into a torus, see Fig. \ref{F:cell} for such a fundamental domain).
\begin{figure}[ht!]
  \centering
  \includegraphics[scale=0.5]{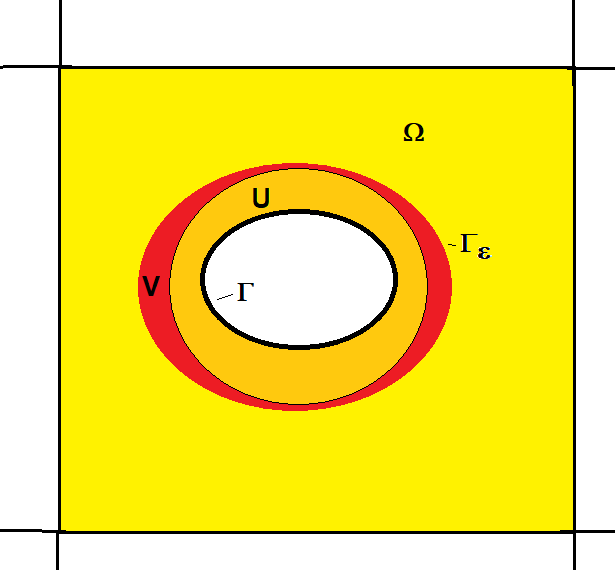}
\caption{The unit cell of the structure.} \label{F:cell}
\end{figure}

We consider a smooth surface $\Gamma_\epsilon$ near to the boundary $\Gamma$ and inside $\Omega$. We split the domain between $\Gamma$ and $\Gamma_\epsilon$ into two nested annular domains $V\subset U$ with another smooth surface approximating $\Gamma$ as a boundary of $V$.  We can  make variations of the hole's boundary $\Gamma$ by considering $C^m$-embeddings $h$ from the annular domain $U$ into $\R^n$ that are close to $I_U$ and coincide with $I_V$ on the sub-annulus $V$. These extend as identity to the whole domain $\Omega$. Using these diffeomorphisms, analogously to what was done in Section \ref{S:perturb}, one can rewrite the eigenvalue problem (\ref{E:disp}) on the $h$-modified domain as a problem on the original domain, but with a modified operator $L(k,D,h)$. As it was explained in Section \ref{S:perturb}, this elliptic operator depends analytically on $h$ and thus Proposition \ref{P:bundle} is applicable with the quasi-momenta $k$ playing the role of parameters $z$. Thus, one obtains the following result:

\begin{theorem}\label{T:disp}
The ``dispersion relation''
\begin{equation}
\mathcal{D}:=\{(\lambda,k,h)\in \C\times\mathcal{C}\times R\,|\,\exists\, u\neq 0 \,\Z^n\mbox{-periodic, } L(k,h)u=\lambda u\}
\end{equation}
is an analytic subset in $\C\times\C^n\times R$. It is principal (i.e., is defined as the set of zeros of a single analytic function $f(\lambda,k,h)$) if the Fredholm index of the operator is equal to zero.

In particular, if $\lambda$ is a simple eigenvalue of the operator $L(k_0)$, then it extends analytically to a simple eigenvalue $\lambda(k,h)$ of $L(k,h)$ for sufficiently small $\|k-k_0\|$ and $\|h-I_U\|$.
\end{theorem}

\section{Absolute continuity of the spectrum}\label{S:ac}

We are now interested in the structure of the spectrum of the periodic elliptic BVP in periodically perforated $\R^n$ (see the previous sections for the exact description).
The standard proof (see \cite{RS,KuchBook}) applies that leads to the absence of singular continuous spectrum:
\begin{theorem}\label{T:nosc}
The singular continuous part of the spectrum of the periodic operator described above is empty.
\end{theorem}
Thus, in proving the absolute continuity of the spectrum the only hurdle is to exclude the possibility of existence of the pure point part of the spectrum. In the generality considered above (an arbitrary order periodic elliptic operator) the statement that the spectrum is absolutely continuous would be incorrect, as the well known examples of elliptic operators of higher order with compactly supported eigenfunctions \cite{plis}(see also usage of this result in the periodic situation in \cite{KuchBAMS,KuchBook}). So, the only realistic option is to restrict ourselves to operators of second order. So, we consider now the Schr\"odinger operator
\begin{equation}\label{E:schr}
-\Delta+V(x),
\end{equation}
where $V(x)$ is a real $\Z^n$-periodic bounded measurable potential. We also impose zero \textbf{Dirichlet conditions} on the boundaries of the perforations.

\begin{theorem}\label{T:ac}
The spectrum of the periodic operator (\ref{E:schr}) with Dirichlet conditions in the perforated domain is absolutely continuous.
\end{theorem}
\begin{proof}
Let us write the operator as
\begin{equation}
\left(\frac{1}{i}\frac{\partial}{\partial x}\right)^2+V(x)
\end{equation}
acting in the natural way in $L_2(W)$ with the domain $H^2_0(W)$. The spectral parameter $\lambda$ can be absorbed by the potential, and thus it is sufficient to consider the case when $\lambda=0$.

The standard Floquet theory (see \cite{KuchBAMS,KuchBook,Skri} and references therein) reduces this operator to the direct integral over the Brillouin zone of the quasimomenta of the operators
\begin{equation}
L(k):=\left(\frac{1}{i}\frac{\partial}{\partial x}+k\right)^2+V(x)
\end{equation}
acting on $\Z^n$-periodic functions on $W$, i.e. on functions on the sub-domain $\widetilde{W}:=W/\Z^n$ of the torus $\T^n=\R^n/\Z^n$. Here $k\in \R^n$ is the quasi-momentum. As it is standard, the family of the Fredholm operators $L(k): H^2_0(\widetilde{W})\mapsto L_2(\widetilde{W})$ extends analytically to the whole $\C^n$. Again, as the standard L.~Thomas' argument shows (see \cite{Thomas,RS,KuchBAMS}), existence of the point $\lambda=0$ in the pure point spectrum of $L$ is equivalent to the following:\\
\textbf{Pure point spectrum condition}: \emph{For any $k\in\C^n$ there exists a non-trivial function $u\in H^2_0(\widetilde{W})$ such that}
\begin{equation}\label{E:test}
L(k)u=0.
\end{equation}

Like in most of the known proofs of absolute continuity, we will try to find a \textbf{complex} value of the quasimomentum
$$k=a+ib,\, a,b\in\R^n$$
for which the equation (\ref{E:test}) has no non-trivial (periodic) solutions.

We rewrite (\ref{E:test}) as follows:
\begin{equation}\label{E:test2}
\left(\frac{1}{i}\frac{\partial}{\partial x}+k\right)^2 u= -V(x) u.
\end{equation}
The $L_2$-norm of the right hand side, due to boundedness of the potential can be estimated for any $u\in L_2(\widetilde{W})$ as follows:
$$
\|Vu\|^2_{L_2}\leq C \|u\|^2_{L_2}.
$$
If now we find a value of $k$ such that for any non-zero $u$ one has
\begin{equation}\label{E:kab}
\|\left(\frac{1}{i}\frac{\partial}{\partial x}+k\right)^2 u\|^2_{L_2} > C\|u\|^2_{L_2},
\end{equation}
the theorem will be proven.

So far, everything went along the standard L.~Thomas' proof. Finding a complex value of $k$ such that (\ref{E:kab}) holds also resembles the case of non-perforated domain, with a little additional caveat, due to $\widetilde{W}$ being only a part of the torus, rather than the whole torus.

So, let $k=a+ib$ and rewrite the differential expression in the left hand side of (\ref{E:kab}) as follows:
\begin{equation}
\left(-\Delta+|a|^2-|b|^2-2ia\cdot \frac{\partial}{\partial x}\right) +2ib\cdot\left(a+i\frac{\partial}{\partial x}\right).
\end{equation}
Let us denote the expression in the first parentheses by $A$ and in the second as $B$, so
in the left hand side of (\ref{E:kab}) is $Au+Bu$. A direct computation of inner products in $L_2(\widetilde{W})$ shows that
\begin{equation}
\|Au+Bu\|^2_{L_2}=\|Au\|^2_{L_2}+\|Bu\|^2_{L_2}.
\end{equation}
The formal reason is that $A$ and $B$ commute, $A$ is symmetric, and $B$ is skew-symmetric. However, the direct calculation works nicely for $u\in H^2_0(\widetilde{W})$, avoiding composition of $A$ and $B$.

So, it is sufficient for a given constant $C>0$ to choose $a,b\in\R^n$ in such a way that
\begin{equation}\label{E;suffic}
\|Bu\|^2_{L_2}>C \|u\|^2_{L_2}
\end{equation}
for any non-zero $u\in H^2_0(\widetilde{W}))$.

Now let us extend $u$ as a function $\tilde{u}$ on the whole torus $\T^n$, setting it equal to zero outside of $\widetilde{W}$. The resulting function is not in $H^2$ on the torus anymore, but is still in $H^1$, due to the Dirichlet condition imposed. Let us expand $\tilde{u}$ into Fourier series:
\begin{equation}
\tilde{u}(x)=\sum\limits_{m\in \Z^n} u_m e^{2\pi ix\cdot m}.
\end{equation}
Let us choose $a=(\alpha,0,\dots,0), b=(\beta,0,\dots,0)$. Then
\begin{equation}\label{E:fourier}
Bu=B\tilde{u}|_{\widetilde{W}}= \sum\limits_{m\in \Z^n} 2i\beta(\alpha-2\pi m_1)u_m e^{ix\cdot m}.
\end{equation}
Choosing now $\alpha=\pi$ and $\beta>C/6$, one sees that the absolute values of all multipliers in the Fourier series (\ref{E:fourier}) exceed $C$, which implies (\ref{E;suffic}), and hence (\ref{E:kab}). This finishes the proof of the theorem. \end{proof}

\section{Remarks and conclusions}\label{S:remarks}
\begin{enumerate}
\item The geometry considered in this article seems to be restricted by two secretly made assumptions: the group of periods being $\Z^n$ (changing variables would lead the Laplacian to look somewhat different), as well as by the perforations being holes in the cells of a cubic tiling of the space. This seems to be excluding perforations like the one in Fig. \ref{F:strange}. In fact, the proof carries out without difficulty to arbitrary Brave lattice and any ``strange'' shape of the fundamental domain instead of a cube. In particular, in the main consideration, any fundamental domain, not necessarily a cube, would fold onto the torus and thus the proof carries through.
\begin{figure}
  \centering
  \includegraphics[scale=0.5]{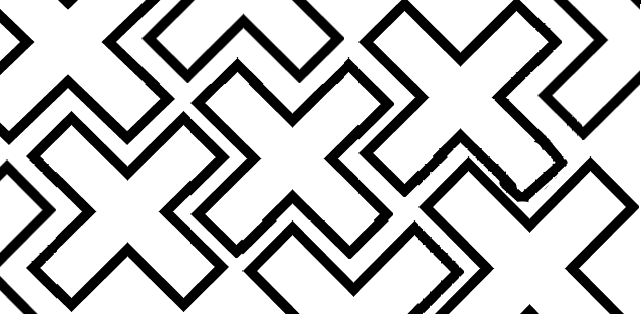}
  \caption{}\label{F:strange}
\end{figure}
\item It was mentioned in \cite{KuchBAMS} that the ``extended'' dispersion relation, i.e. including for instance potentials, is also analytic. This, however did not include the domain shape variations. We show that these can also be included, which in particular implies analyticity of an isolated eigenvalue. However, the global analyticity of the dispersion relation carries more information, including also its structure near non-smooth points as well. It is also known to be important for other issues of spectral theory, see e.g. \cite{Colin,Gerard,Gies,Karp,Knor,KuchBAMS,KuchBook,KuchVain1,KucVainb}.
\item
One expects that the statement of Theorem \ref{T:ac} about absolute continuity of the spectrum holds for any second order periodic elliptic operator with sufficiently ``nice'' coefficients, instead of just the Schr\"odinger operator (\ref{E:schr}). This is known to be not a simple problem even without perforations (see \cite{BirSus,Fr,KuchBAMS,KuchBook,KuLev,Morame,Shen,Sobolev,SobWal,Thomas}). The famous work by Friedlander \cite{Fr} proves absolute continuity under an additional symmetry condition that the operator should be even with respect to one of the coordinates corresponding to the periodicity axes. This proof applies without any modification to the perforated domain case, if the symmetry condition is imposed on the shape of the perforation as well:
\begin{theorem}\label{T:fr}
Let $(x_1,...,x_n)$ be the coordinates in $\R^n$, such that the natural generators of the group of periods $\Z^n$ act by shifting the appropriate coordinates.

Suppose that the periodic 2nd order elliptic operator $L$ in the perforate space, equipped with Dirichlet boundary conditions, as well as the shape of the perforation are symmetric w.r.t. the transformation $(x_1,\dots,x_n) \mapsto (-x_1,\dots,x_n)$. Then the spectrum is absolutely continuous.
\end{theorem}
\end{enumerate}

\section*{Acknowledgments}
P.K. acknowledges support of the NSF DMS grants \# 1517938 and \# 2007408. JT was partially supported by a research grant from the Faculty of Science of the University of Helsinki. Thanks also got to Professors L.~Friedlander, M.~Lanza de Cristoforis, and P.~Musolino for discussions

\end{document}